\documentclass[reqno,11pt]{amsart}
\usepackage{amscd,amssymb,verbatim,mathrsfs}

\usepackage{ifthen}

\makeatletter

\renewcommand\@biblabel[1]{#1.}
\makeatother
\theoremstyle{plain}

\newtheorem{theorem}{Theorem}[section]
\newtheorem{defn}{Definition}[section]
\newtheorem{rem}{Remark}[section]

\newtheorem{Lem}{Lemma}[section]
\newtheorem{Prop}{Proposition}[section]
\newtheorem{Cor}{Corollary}[section]

\newtheorem{notation}{Notation}[section]

          { \begin{Thm}  }%
          { \end{Thm} }

\newenvironment{lem}%
          { \begin{Lem}    }%
          { \end{Lem} }

          { \begin{Prop}  }%
          { \end{Prop} }

          { \begin{Cor} }%
          { \end{Cor} }

          { \begin{Def} }%
          { \end{Def} }

\newcommand{\lw}{\Gamma_{\ell_{m}^2}(X,F)}
\newcommand{\lpw}{\Gamma_{\ell_{m}^p}(X,F)}
\newcommand{\lqw}{\Gamma_{\ell_{m}^{p^*}}(X,F)}

\newcommand{\lwp}{{\ell^{p}_{m}}(X)}
\newcommand{\lwi}{\Gamma_{\ell^{\infty}}(X,F)}

\newcommand{\Dom}{\operatorname{Dom}}

\newcommand{\hpm}{H_{p,\max}}
\newcommand{\hpmin}{H_{p,\min}}

\title{Maximal Accretive Extensions of Schr\"odinger Operators on Vector Bundles over Infinite Graphs}
\author{Ognjen Milatovic}
\address{Department of Mathematics
and Statistics \\ University of North Florida \\ Jacksonville, FL
32224 \\ USA}
\email{omilatov@unf.edu}
\author{Fran{\c c}oise Truc}
\address{Grenoble University\\ Institut Fourier\\
Unit{\'e} mixte de recherche CNRS-UJF 5582\\
BP 74\\
38402-Saint Martin d'H\`eres Cedex
\\France}

\email{francoise.truc@ujf-grenoble.fr}

\begin{document}
\maketitle

\begin{abstract}
Given a Hermitian vector bundle over an infinite weighted graph,  we define the Laplacian associated to a unitary connection on this bundle and study a perturbation of this Laplacian by an operator-valued potential. We give a sufficient condition for the resulting Schr\"odinger operator to serve as the generator of a strongly continuous contraction semigroup in the corresponding $\ell^{p}$-space. Additionally, in the context of $\ell^{2}$-space, we study the essential self-adjointness of the corresponding Schr\"odinger operator.
\end{abstract}

\section{Introduction}\label{S:main}
In recent years, there has been quite a bit of interest in the study  of the Laplacian in $\ell^{p}$-spaces on infinite graphs.
More precisely, let  $(X,b,m)$ be a weighted graph as described in section~\ref{SS:setting} below,
and  let us define a form $Q^{(c)}$  on
(complex-valued) finitely supported functions on $X$ by
\begin{align}\label{E:form-comp}
Q^{(c)}(u,v)&:=\frac{1}{2}\sum_{x,y\in X}b(x,y)(u(x)- u(y))(\overline{v(x)- v(y)})\nonumber\\
  &+\sum_{x\in X}w(x)u(x)\overline{v(x)},
\end{align}
where $w\colon X\to [0,\infty)$.
We denote by $\ell^p_{m}(X)$ the space of $\ell^p$-summable functions with weight $m$, by $Q^{(D)}$ the closure of $Q^{(c)}$
in $\ell^2_{m}(X)$, and by $L$ the associated self-adjoint operator. Since $Q^{(D)}$ is a Dirichlet form,
the semigroup $e^{-tL}$, $t\geq 0$, extends to a $C_0$-semigroup on $\lwp$, where $p\in [1,\infty)$.  We denote by $-L_p$ the generators of these semigroups. For the definition of a $C_0$-semigroup and its generator, see the Appendix. The following characterization of operators $L_p$ is given in~\cite{Keller-Lenz-09}:

\medskip

\noindent\emph{Assume that
\[
\sum_{n\in {\mathbb{Z}}_{+}}m(x_n)=\infty , \eqno{{\rm \bf{ (A1)}}}
\]
for any sequence $\{x_n\}_{n\in{\mathbb{Z}}_{+}}$ of vertices such that  $x_{n}\sim x_{n+1}$ for all $n\in{\mathbb{Z}}_{+}$.
Then for any $p\in[1,\infty)$, the operator $L_p$ is the restriction of $\widetilde{L}$
to
\[
\Dom(L_p) = \{u\in\lwp\cap \widetilde{D_s}:\, \widetilde{L}u\in\lwp\},
\]
where
\[
\widetilde{D_{s}}:=\{u\colon X\to\mathbb{C}\colon\, \displaystyle\sum_{y\in X}b(x, y)|u(y)|<\infty, \forall x\in X\},
\]
$\widetilde{L}:=\Delta_{b,m}+w/m$, and
\begin{equation}\label{E:ord-lap}
(\Delta_{b,m} u)(x):=\frac{1}{m(x)}\sum_{y\in X}b(x,y)(u(x)-u(y)).
\end{equation}}

Actually, (A1) can be replaced  when $w=0$ by the existence of a compatible intrinsic metric (see~\cite{Hua-Keller-13}), or
if moreover $p=2$, by the existence of an intrinsic metric so that $\frac{1}{m(x)}\sum_{y\in X}b(x,y)$ is bounded on the combinatorial neighborhood of each distance ball (see~\cite{HKMW}).

In the case of Schr\"odinger operators on  a Riemannian manifold $M$, it is natural to study maximal accretivity or self-adjointness  properties
of operators acting on sections of vector bundles over $M$.
But the notion of vector bundle is also relevant on graphs; see for example
~\cite{FKC-S-92}, ~\cite{GMT-13}, ~\cite{Kenyon-11}, and~\cite{Si-Vu-12}.
The aim of this paper is precisely to study such properties in the setup of a vector bundle over an infinite weighted graph.
In particular, we give sufficient conditions
for the equality of the operator $\hpm$ (vector-bundle analogue of $L_p$) and the closure in $\lpw$ (the corresponding
$\ell^{p}$-space of sections of the bundle $F\to X$) of the restriction of $\widetilde{H}_{W,\Phi}$ (vector-bundle analogue of
$\widetilde{L}$) to the set of finitely supported sections. 

The paper is organized as follows. 
In sections~\ref{SS:setting},~\ref{SS:v-bundle} and ~\ref{SS:operator} we describe
the setting: discrete sets, Hermitian vector bundle and connection,  operators.
The main results are presented in section~\ref{SS:theorems}, with some comments.
Section~\ref{S:preli} contains  preliminary results, such as Green's formula, Kato's inequality, and ground state transform. Sections~\ref{S:kato-lp-1},~\ref{SS:main-1-1} and \ref{S:t-5} are devoted to  the proofs of the theorems. For readers' convenience, in the Appendix we review some concepts from the theory of semigroups of operators: $C_0$-semigroup, generator of a $C_0$-semigroup, and (maximal) accretivity. Additionally, the Appendix contains the statement of Hille--Yosida Theorem and a discussion of the connection between self-adjointness and maximal accretivity of operators in Hilbert spaces.

\section{Setup and Main Results}\label{S:results}
\subsection{Weighted Graph}\label{SS:setting}
Let $X$ be a countably infinite set, equipped with a measure $m\colon X\to (0,\infty)$.
Let $b\colon X\times X\to[0,\infty)$ be a function such that

\medskip

\noindent (i) $b (x, y) = b (y, x)$, for all $x,\,y\in X$;

\medskip

\noindent (ii) $b(x,x)=0$, for all $x\in X$;

\medskip

\noindent (iii) $\displaystyle\sum_{y\in X}b(x,y)<\infty$, for all $x\in X$.

\noindent Vertices $x,\, y\in X$ with $b(x, y) > 0$ are called \emph{neighbors}, and we denote this relationship by $x\sim y$. We call the triple $(X,b,m)$ a \emph{weighted graph}. We assume that $(X,b,m)$ is connected, that is, for any $x,\,y\in X$ there exists a path $\gamma$ joining $x$ and $y$. Here, a path $\gamma$ is a sequence $x_1,\,x_2,\,\dots,x_n\in X$ such that $x=x_1$, $y=x_n$, and $x_{j}\sim x_{j+1}$ for all $1\leq j\leq n-1$.

\subsection{Hermitian Vector Bundles on Graphs and Connection}\label{SS:v-bundle} A family of (finite-dimensional) complex
linear spaces $F=\bigsqcup_{x\in X}F_x$ is called a \emph{complex vector bundle over $X$} and written $F\to X$, if any two
$F_x$ and $F_y$ are isomorphic as complex vector spaces. Then the $F_x$\rq{}s are called the \emph{fibers} of $F\to X$, and
the complex linear space
$$
\Gamma(X,F):=\prod_{x\in X} F_x=\left\{u|\, u\colon X\to F,\, u(x)\in F_x\right\}
$$
is called the space of \emph{sections in $F\to X$}. We define the space of \emph{finitely supported sections} $\Gamma_c(X,F)$ of
 $F\to X$ as the set of $u\in \Gamma(X,F)$ such that $u(x)=0$ for all but finitely many $x\in X$.

\begin{defn}  An assignment $\Phi$ which associates to any $x\sim y$ an isomorphism of complex vector spaces
 $\Phi_{x,y}\colon F_x\to F_y$ is called a connection on the complex vector bundle $F\to X$ if
\begin{equation}\label{E:phi-transp}
\Phi_{y,x}=\left(\Phi_{x,y}\right)^{-1}\,\qquad \textrm{for all }x\sim y.
\end{equation}
\end{defn}

\begin{defn}\label{herm} (i) A family of complex scalar products
$$
\langle\cdot,\cdot\rangle_{F_x}\colon F_x\times F_x\to \mathbb{C}, \quad x\in X,
$$
is called a Hermitian structure on the complex vector bundle $F\to X$, and the pair given by $F\to X$ and $\langle\cdot,\cdot\rangle_{F_x}$ is called a Hermitian vector bundle over $X$.

\medskip

\noindent (ii) A connection $\Phi$ on a complex vector bundle $F\to X$ is called unitary
with respect to a Hermitian structure $\langle\cdot,\cdot\rangle_{F_x}$ if for all $x\sim y$ one has
\begin{equation}\label{E:phi-unit}\nonumber
\Phi_{x,y}^*=\Phi^{-1}_{x,y},
\end{equation}
where $T^*$ denotes the Hermitian adjoint of an operator $T\colon F_x\to F_y$ with respect to $\langle\cdot,\cdot\rangle_{F_x}$ and $\langle\cdot,\cdot\rangle_{F_y}$.
\end{defn}

\begin{defn}\label{lapl} The  \emph{Laplacian} $\Delta^{F,\Phi}_{b,m}\colon \widetilde{D}\to \Gamma(X,F)$ on a Hermitian vector
bundle $F\to X$ with a unitary connection $\Phi$ is a linear operator with the domain
\begin{equation}\label{E:dom-F}
\widetilde{D}:=\{u\in \Gamma(X,F)\colon \displaystyle\sum_{y\in X}b(x, y)|u(y)|_{F_{y}} <\infty,\textrm{ for all }x\in X\}
\end{equation}
defined by the formula
\begin{equation}\label{E:magnetic-lap}
    (\Delta^{F,\Phi}_{b,m} u)(x)=\frac{1}{m(x)}\sum_{y\in X}b(x,y)(u(x)-\Phi_{y,x} u(y)).
\end{equation}
\end{defn}
\begin{rem} \emph{The operator $\Delta^{F,\Phi}_{b,m}$ is well-defined by the property (iii) of $b(x,y)$, definition~(\ref{E:dom-F}), and unitarity of $\Phi$.}
\end{rem}
\begin{rem} \emph{In the case $F_x=\{x\}\times\mathbb{C}$ with the canonical Hermitian structure, the sections
of the bundle $F\to X$ can be canonically identified with complex-valued functions on $X$. Under this
identification, any connection $\Phi$ can be uniquely written as $\Phi_{x,y}=e^{i\theta(y,x)}$,
 where $\theta \colon X\times X\to [-\pi,\pi]$ is a magnetic potential on $(X,b)$, which,
due to~(\ref{E:phi-transp}), satisfies the property $\theta(x,y)=-\theta(y,x)$ for all $x,\,y\in X$. As a result, we get
 the magnetic Laplacian operator. In particular, if $\theta\equiv 0$ we get the Laplacian operator~(\ref{E:ord-lap}).}
\end{rem}
\begin{rem} \emph{If the property (iii) of $b(x,y)$ is replaced by}
\[
\displaystyle\sharp\,\{y\in X\colon b(x,y)>0\}<\infty, \textrm{ for all }x\in X,
\]
\emph{where $\sharp\, S$ denotes the number of elements in the set $S$, then the graph $(X,b,m)$ is called locally finite. In this case, we have $\widetilde{D}=\Gamma(X,F)$.}
\end{rem}

\subsection{Operators} \label{SS:operator}
From now on  we will always work in the setting of a Hermitian vector bundle $F\to X$ over a connected weighted graph $(X, b, m)$,
equipped with a unitary connection $\Phi$.


\begin{defn}\label{schto} We define the Schr\"odinger-type operator $\widetilde{H}_{W,\Phi}\colon \widetilde{D}\to \Gamma(X,F)$ by the formula
\begin{equation}\label{E:magnetic-schro}
\widetilde{H}_{W,\Phi}u:= \Delta^{F,\Phi}_{b,m} u +Wu,
\end{equation}
where $W(x)\colon F_{x}\to F_{x}$ is a linear operator for any  $x\in X$, and $\widetilde{D}$ is as in~(\ref{E:dom-F}).
\end{defn}

\begin{defn}\label{D:l-p-space} (i)
For any  $1\leq p<\infty$ we denote by $\lpw$  the space of sections $u\in \Gamma(X,F)$ such that
\begin{equation}\label{E:l-p-def}\nonumber
\|u\|_{p}^{p}:=\sum_{x\in X}m(x)|u(x)|^p_{F_x} <\infty,
\end{equation}
where $|\cdot|_{F_x}$ denotes the norm in $F_{x}$ corresponding to the Hermitian product $\langle \cdot,\cdot\rangle_{F_{x}}$.
The space of $p$-summable functions $X\to\mathbb{C}$ with weight $m$ will be denoted by $\ell^p_{m}(X)$.

\noindent (ii) By ${\lwi}$ we denote the space of bounded sections of $F$, equipped with the norm
\begin{equation}\label{E:l-infty-def}\nonumber
\|u\|_{\infty}:=\sup_{x\in X}|u(x)|_{F_x}.
\end{equation}
The space of bounded functions on $X$ will be denoted by $\ell^{\infty}(X)$.
\end{defn}

The space $\lw$ is a Hilbert space with the inner product
\begin{equation}\label{E:inner-w}\nonumber
(u,v):=\sum_{x\in X}m(x)\langle u(x),v(x)\rangle_{F_{x}}
\end{equation}
\begin{defn}\label{D:max-op} Let $1\leq p<+\infty$ and let $\widetilde{D}$ be as in~(\ref{E:dom-F}).  The maximal operator $\hpm$ is given by the formula $\hpm u=\widetilde{H}_{W,\Phi}u$ with
domain
\begin{equation}\label{E:domhpm}
\Dom(\hpm)=\{u\in \lpw\cap \widetilde{D}: \widetilde{H}_{W,\Phi} u\in \lpw\}.
\end{equation}
Moreover if
\begin{equation}\label{E:assumption-vcomp-lwp}
\widetilde{H}_{W,\Phi}[\Gamma_{c}(X,F)]\subseteq \lpw,
\end{equation}
then we set $\hpmin:=\widetilde{H}_{W,\Phi}|_{\Gamma_{c}(X,F)}$.\end{defn}
\begin{rem}\label{R:min-def-op} \emph{Note that under our assumptions on $(X,b,m)$, the inclusion (\ref{E:assumption-vcomp-lwp}) does not necessarily hold.
It holds if we additionally assume that $(X,b,m)$ is locally finite.}
\end{rem}

\subsection{Statement of the Results}\label{SS:theorems}
Let us denote by $\overline{T}$ the closure of an operator $T$.

\begin{theorem}\label{T:main-1} Let $W(x)\colon F_{x}\to F_{x}$ be a linear operator satisfying
\begin{equation}\label{E:bdd-below-w}
\textrm{Re}\,\langle W(x)u(x),u(x)\rangle_{F_{x}}\geq 0, \qquad\textrm{for all }x\in X.
\end{equation}
Then, the following properties hold:
\begin{enumerate}
\item [(i)]  Let $1<p<\infty$, and assume that~(\ref{E:assumption-vcomp-lwp}) and (A1) are satisfied. Then the operator $-\overline{\hpmin}$ generates a strongly continuous contraction semigroup on $\lpw$.
\item [(ii)] Assume that~(\ref{E:assumption-vcomp-lwp}) is satisfied for $p=1$, and that $(X, b,m)$ is stochastically complete. Then the operator $-\overline{{H_{1,\min}}}$ generates a strongly continuous contraction semigroup on $\Gamma_{\ell^1_{m}}(X,F)$.
\end{enumerate}
\end{theorem}
\begin{rem}\label{R:stochastic-completeness}
\emph{By Definition 1.1 in~\cite{Keller-Lenz-09}, stochastic completeness of $(X,b,m)$ means that there is no non-trivial and non-negative $w\in \ell^{\infty}(X)$ such that}
\[
(\Delta_{b,m}+\alpha)w\leq 0, \qquad\alpha>0,
\]
\emph{where $\Delta_{b,m}$ is as in~(\ref{E:ord-lap}).}
\end{rem}
\begin{rem}\label{R:max-accretive} \emph{The notions of generator of a strongly continuous semigroup and (maximal) accretivity are reviewed in the Appendix. In particular, under the assumptions of Theorem~\ref{T:main-1}, the operator $\overline{\hpmin}$ is maximal accretive for all $1\leq p<\infty$.}
\end{rem}
In the next theorem, we make the following assumption, which is stronger than~(\ref{E:assumption-vcomp-lwp}):
\begin{equation}\label{E:assumption-vcomp-lwp-q}
\widetilde{H}_{W,\Phi}[\Gamma_{c}(X,F)] \subseteq\lpw\cap\lqw,
\end{equation}
with $1/p+1/p^*=1$.
\begin{rem}\label{R:well-def-op-new} \emph{If $(X,b,m)$ is a locally finite graph then~(\ref{E:assumption-vcomp-lwp-q}) is satisfied. If $\inf_{x\in X}m(x)>0$ then (A1) and~(\ref{E:assumption-vcomp-lwp-q}) are satisfied.}
\end{rem}

\begin{theorem}\label{T:main-1-1}  Assume that the hypotheses (A1) and (\ref{E:bdd-below-w}) are satisfied. Then, the following properties hold:
\begin{enumerate}
\item [(i)] Let $1<p<\infty$, and assume that~(\ref{E:assumption-vcomp-lwp-q}) is satisfied. Then $\overline{\hpmin}=\hpm$.
\item [(ii)] Assume that~(\ref{E:assumption-vcomp-lwp-q}) is satisfied for $p=1$, and that $(X, b,m)$ is stochastically complete. Then $\overline{{H_{1,\min}}}=H_{1,\max}$.
\end{enumerate}
\end{theorem}


\noindent Regarding self-adjointness problems, let us point out that the results of~\cite{vtt-11-4, vtt-11-3, Masamune-09, Milatovic-11,
Milatovic-Truc-13} and Theorem 5 in~\cite{Keller-Lenz-10} can be extended to the vector-bundle setting. As an illustration, we state
and prove an extension of Theorem 1.5 from \cite{Milatovic-Truc-13}. Before doing this, we recall the notion of intrinsic metric.

\begin{defn}  A pseudo metric  is a map $d\colon X\times X \to [0,\infty)$  such that $d(x,y)=d(y,x)$, for all $x,\,y\in X$; $d(x,x)=0$,
 for all $x\in X$; and $d(x,y)$ satisfies the triangle inequality. \\
A pseudo metric $d=d_{\sigma}$ is called a path pseudo metric if there exists a map $\sigma\colon X\times X\to [0,\infty)$ such that
$\sigma(x,y)=\sigma(y,x)$, for all $x,\,y\in X$; $\sigma(x,y)>0$ if and only if $x\sim y$; and
$
d_{\sigma}(x,y)=\inf\{l_{\sigma}(\gamma):\gamma \textrm{ path connecting }x\textrm{ and }y\},
$
where the length $l_{\sigma}$ of the path $\gamma=(x_0,x_1,\dots,x_n)$ is given by
\begin{equation}\label{E:l-sigma-def}\nonumber
l_{\sigma}(\gamma)=\sum_{i=0}^{n-1}\sigma(x_i,x_{i+1}).
\end{equation}
\end{defn}
On a locally finite graph a path pseudo metric is a metric; see~\cite{HKMW}.
\begin{defn}\label{D-intrinsic} A pseudo metric $d$ on $(X,b,m)$ is called intrinsic if
\[
\frac{1}{m(x)}\sum_{y\in X}b(x,y)(d(x,y))^2\leq 1,\qquad\textrm{for all }x\in X.
\]
\end{defn}
\begin{rem}
\emph{The concept of intrinsic pseudo metric  goes back to~\cite{flw} which discusses a more general situation. For graphs it has been discussed in
~\cite{Huang-11} and ~\cite{folz-11}. Related earlier material can be found in~\cite{MU-11}.}
\end{rem}
We will also use the notion of a \emph{regular graph} introduced in~\cite{vtt-11-4}, which is a (not yet published) revised version of
~\cite{vtt-11-2}. Let us first recall the definition of the boundary of a given set $A\subseteq X$:
\[
\partial A:=\{x\in A\colon \textrm{there exists }y\in X\backslash A\textrm{ such that }y\sim x\}.
\]
In the sequel, we denote by $(\widehat{X},\widehat{d})$ the metric completion of $(X,d)$, and we define the \emph{Cauchy boundary}
$X_{\infty}$ as follows: $X_{\infty}:=\widehat{X}\backslash X$. Note that $(X,d)$ is metrically complete if and only if $X_{\infty}$
 is empty. For a path metric $d=d_{\sigma}$ on $X$ and $x\in X$, we set
\begin{equation}\label{E:dist-boundary}
D(x):=\inf_{z\in X_{\infty}}\widehat{d}_{\sigma}(x,z).
\end{equation}
 \begin{defn}\label{ikx} Let $(X,b,m)$ be a graph with a path metric ${d}_{\sigma}$. Let $\varepsilon>0$ be given and let
\begin{equation}\label{E:x-eps}
X_{\varepsilon}:=\{ x\in X \colon D(x) \geq \varepsilon\}.
\end{equation}
We say that  $(X,b,m)$  is \emph{regular} if for any sufficiently small $\varepsilon$, any bounded
subset of $\partial X_{\varepsilon}$ (for the metric $d_{\sigma}$)  is finite.
\end{defn}
\begin{rem} \emph{Metrically complete graphs $(X,d)$ are regular since $D(x) = \infty$ for any $x \in X$,
which implies that  $X_{\varepsilon} = X$, so that $\partial X_{\varepsilon}= \emptyset$.}
\end{rem}
\begin{rem}
\emph{Definition~\ref{ikx} covers also a broad class of metrically non-complete graphs.  For instance, weighted graphs whose first Betti number is finite are regular. In particular, any weighted tree is  regular; see \cite{vtt-11-4}.}
\end{rem}

\begin{theorem}\label{T:main-5}  Let $(X, b, m)$ be a locally finite graph with an intrinsic path metric  $d=d_{\sigma}$. Assume that  $(X, b, m)$
is regular. Let $W(x)\colon F_{x}\to F_{x}$ be a linear self-adjoint operator such that there exists a constant $C$ satisfying
\begin{equation}\label{E:potential-minorant}
\langle W(x)u(x),u(x)\rangle_{F_{x}} \geq\left(\frac{1}{2(D(x))^2}-C\right)|u(x)|^2_{F_x},
\end{equation}
for all $x\in X$ and all $u \in \Gamma_{c}(X,F)$, where $D(x)$ is as in~(\ref{E:dist-boundary}).
Then $\widetilde{H}_{W,\Phi}$ is essentially self-adjoint on $\Gamma_{c}(X,F)$.
\end{theorem}

\section{Preliminary Lemmas}\label{S:preli}

\subsection{Green's Formula}
We now give a variant of Green's formula, which is analogous to Lemma 2.1 in~\cite{GKS-13} and Lemma 4.7 in~\cite{HK-2011}.
\begin{notation}\label{adj}
 Let $W(x)\colon F_{x}\to F_{x}$ be a linear operator. We denote by $W^*$  the Hermitian adjoint of $W$, that is, $(W(x))^*$ is the Hermitian adjoint of $W(x)$ with respect to $\langle\cdot,\cdot\rangle_{F_{x}}$.
\end{notation}

\begin{lem}\label{L:Green}  Let  $\widetilde{H}_{W,\Phi}$ be as in~(\ref{E:magnetic-schro}). The following properties hold:
\begin{itemize}
\item [(i)] if $\widetilde{H}_{W,\Phi}[\Gamma_{c}(X,F)]\subseteq\lpw$ for some $1\leq p\leq\infty$, then any $u\in\lqw$ with $1/p+1/p^*=1$ belongs to the set $\widetilde{D}$ defined by~(\ref{E:dom-F});
\item[(ii)] for all $u\in \widetilde{D}$ and all $v\in \Gamma_{c}(X,F)$, the sums
\[
\sum_{x\in X}m(x)\langle \widetilde{H}_{W,\Phi}u,v\rangle_{F_x},\qquad \sum_{x\in X}m(x)\langle u,\widetilde{H}_{W^*,\Phi}v\rangle_{F_x},
\]
 and the expression
\begin{align}\label{E:form-comp-sum}
&\frac{1}{2}\sum_{x,y\in X}b(x,y)\langle u(x)-\Phi_{y,x} u(y),v(x)-\Phi_{y,x} v(y)\rangle_{F_x}\nonumber\\
&+\sum_{x\in X}m(x)\langle W(x)u(x),v(x)\rangle_{F_x}
\end{align}
converge absolutely and agree.
\end{itemize}
\end{lem}
\begin{proof} To make the notations simpler, throughout the proof we suppress $F_{x}$ in $|\cdot|_{F_x}$. From the assumption $\widetilde{H}_{W,\Phi}[\Gamma_{c}(X,F)]\subseteq \lpw$, it is easily seen that the function $y\mapsto b(x, y)/m(y)$ belongs to $\ell_{m}^{p}(X)$, for all  $x\in X$. In the case $1<p^*<\infty$, for all $u\in\lqw$, by H\"older's inequality with $1/p+1/p^*=1$ we have
\begin{align*}\label{E:quad-form-convergence}
&\sum_{y\in X}b(x,y)|u(y)|\leq \left(\sum_{y\in X}\left(\frac{b(x,y)}{m(y)}\right)^{p}m(y)\right)^{1/p}\left(\sum_{y\in X}|u(y)|^{p^*}m(y)\right)^{1/p^*}.
\end{align*}
In the case $p^*=1$, for all $u\in \Gamma_{\ell_{m}^{1}}(X,F)$, by H\"older's inequality with $p=\infty$ and $p^*=1$ we have
\begin{equation}\label{E:quad-form-convergence-1}\nonumber
\sum_{y\in X}b(x,y)|u(y)|\leq \sup_{y\in X}\left(\frac{b(x,y)}{m(y)}\right)\left(\sum_{y\in X}|u(y)|{m(y)}\right).
\end{equation}
In the case $p^*=\infty$, for all $u\in \Gamma_{\ell^{\infty}}(X,F)$, by H\"older's inequality with $p=1$ and $p^*=\infty$ we have
\begin{equation}\label{E:quad-form-convergence-2}\nonumber
\sum_{y\in X}b(x,y)|u(y)|\leq \sup_{y\in X}\left(|u(y)|\right)\left(\sum_{y\in X}b(x,y)\right).
\end{equation}
This concludes the proof of property (i). Let us prove  property (ii).
Since $v\in \Gamma_{c}(X,F)$, the first sum is performed over finitely many $x\in X$. Hence, this sum converges absolutely. The proof of absolute convergence of the second sum and the expression~(\ref{E:form-comp-sum}) is based on the next two estimates. By Cauchy--Schwarz inequality and unitarity of $\Phi_{y,x}$ we get
\begin{align*}
&\sum_{x,y\in X}|b(x,y)\langle u(x),\Phi_{y,x} v(y)\rangle_{F_x}|\leq
\sum_{y\in X}|v(y)|\left(\sum_{x\in X}b(x,y)|u(x)|\right)<\infty,
\end{align*}
where the convergence follows from the fact that $u\in \widetilde{D}$ and $v\in \Gamma_{c}(X,F)$. Similarly,
\begin{align*}
&\sum_{x,y\in X}|b(x,y)\langle u(x), v(x)\rangle_{F_x}|\leq
\sum_{x\in X}|u(x)||v(x)|\left(\sum_{y\in X}b(x,y)\right)<\infty,
\end{align*}
where the convergence follows by property (iii) of $b(x,y)$ and since $v\in \Gamma_{c}(X,F)$.
The equality of the three sums follows directly from Fubini's theorem. This shows property (ii).
\end{proof}

\subsection{Kato's Inequality} This version of Kato's inequality extends that of~\cite{Dodziuk-Mathai-03}. 

\begin{lem}\label{L:kato-discrete} Let $\Delta_{b,m}$ and $\Delta^{F,\Phi}_{b,m}$ be defined as in~(\ref{E:ord-lap}) and~(\ref{E:magnetic-lap}) respectively. Then,  the following pointwise inequality holds for all $u\in \widetilde{D}$:
\begin{equation}\label{E:kato-discrete}
|u|(\Delta_{b,m}|u|)\leq \textrm{Re}\,\langle \Delta^{F,\Phi}_{b,m} u,u\rangle_{F_x},
\end{equation}
where $|\cdot|$ denotes the norm in $F_{x}$, and $\textrm{Re}\,z$ denotes the real part of a complex number $z$.
\end{lem}
\begin{proof}
Using~(\ref{E:ord-lap}), ~(\ref{E:magnetic-lap}), and the unitarity of $\Phi_{y,x}$, we obtain\\
$
|u(x)|((\Delta_{b,m}|u|)(x))-\textrm{Re}\,\langle \Delta^{F,\Phi}_{b,m} u(x),u(x)\rangle_{F_x}\\
=\frac{1}{m(x)}\sum_{y\in X}b(x,y)\left[
\textrm{Re}\,\langle \Phi_{y,x}u(y),u(x)\rangle_{F_x}-|u(x)||u(y)|\right]\leq 0.
$
\end{proof}
\subsection{Ground State Transform}
Using the definition of $\widetilde{H}_{W,\Phi}$ and unitarity of $\Phi_{y,x}$, it is easy to prove the following vector-bundle analogue of ``ground state transform" from~\cite{flw}, ~\cite{GKS-13}, and~\cite{HK-2011}. We omit the proof here.
\begin{lem}\label{L:lemma-ground-state}   Assume that $W(x)\colon F_{x}\to F_{x}$ is a self-adjoint operator. Assume that~(\ref{E:assumption-vcomp-lwp}) is satisfied for $p=2$. Let $\lambda\in\mathbb{R}$, and  let $u\in\widetilde{D}$ so that
\begin{equation}\nonumber
(\widetilde{H}_{W,\Phi}-\lambda)u=0.
\end{equation}
Then, for all finitely supported functions $g\colon X\to\mathbb{R}$, we have
\begin{align*}\nonumber
&((\widetilde{H}_{W,\Phi}-\lambda)(gu),gu)=\frac{1}{2}\sum_{x,y\in X}b(x,y)(g(x)-g(y))^2(\textrm{Re}\,\langle u(x),\Phi_{y,x} u(y)\rangle_{F_x}).
\end{align*}
\end{lem}

\section{Proof of Theorem~\ref{T:main-1}}\label{S:kato-lp-1}
In Lemmas~\ref{L:hpmin-accretive} and~\ref{L:hpmin-dense-range} below, we assume that the hypotheses of Theorem~\ref{T:main-1} are satisfied.

\begin{lem}\label{L:hpmin-accretive} Let $1\leq p<\infty$. Then, the operator $\hpmin$ satisfies the following
inequality for all $u\in \Gamma_{c}(X,F)$:
\begin{equation}\label{E:accretive-hpmin}
\textit{Re }\sum_{x\in X}  m(x)\langle (\hpmin u)(x),u(x)|u(x)|^{p-2}\rangle_{F_x}\geq 0.
\end{equation}
\end{lem}
\begin{proof} Let $u\in \Gamma_{c}(X,F)$ be arbitrary. By Lemma~\ref{L:Green}(ii) with $W=0$, $u\in \Gamma_{c}(X,F)$ and $v:=u|u|^{p-2}$, we have
\begin{align}\label{E:young-accretive}
&\textrm{Re}\sum_{x\in X}m(x)\langle (\Delta^{F,\Phi}_{b,m}u)(x),u(x)|u(x)|^{p-2}\rangle_{F_{x}}=\frac{1}{2}\sum_{x,y\in X}b(x,y)\left[|u(x)|^{p}\right.\nonumber\\
&\left.+|u(y)|^{p}-\textrm{Re}\langle \Phi_{y,x}u(y),u(x)|u(x)|^{p-2}\rangle_{F_{x}}-\textrm{Re}\langle \Phi_{x,y}u(x),u(y)|u(y)|^{p-2}\rangle_{F_y}\right]\nonumber\\
&\geq
\frac{1}{2}\sum_{x,y\in X}b(x,y)\left[|u(x)|^{p}+|u(y)|^{p}-|u(x)||u(y)|^{p-1}\right.\nonumber\\
&\left.-|u(y)||u(x)|^{p-1}\right].
\end{align}

For $p=1$, from~(\ref{E:young-accretive}) and the assumption~(\ref{E:bdd-below-w}) we easily get~(\ref{E:accretive-hpmin}).

Let $1<p<\infty$ and let $p^*$ satisfy $1/p+1/p^*=1$. By Young's inequality we have
\begin{equation}\label{E:young-1}\nonumber
|u(x)||u(y)|^{p-1}\leq \frac{|u(x)|^{p}}{p}+\frac{(|u(y)|^{p-1})^{p^*}}{p^*}=\frac{|u(x)|^{p}}{p}+\frac{(p-1)|u(y)|^{p}}{p}
\end{equation}
and, likewise,
\begin{equation}\label{E:young-2}\nonumber
|u(y)||u(x)|^{p-1}\leq \frac{|u(y)|^{p}}{p}+\frac{(p-1)|u(x)|^{p}}{p}.
\end{equation}

From the last two inequalities we get
\begin{equation}\label{E:young-3}
-|u(x)||u(y)|^{p-1}-|u(y)||u(x)|^{p-1}\geq -|u(x)|^{p}-|u(y)|^{p}.
\end{equation}
Using~(\ref{E:young-3}),~(\ref{E:young-accretive}), and the assumption~(\ref{E:bdd-below-w}), we obtain~(\ref{E:accretive-hpmin}).
\end{proof}

The following lemma is a special case of Proposition 8 in~\cite{Keller-Lenz-10}:
\begin{lem}\label{L:keller-lenz-inequality} Assume (A1). Let $\alpha > 0$ and $1\leq p<\infty$. Let $\Delta_{b,m}$ be as in~(\ref{E:ord-lap}). Assume that $u\in \ell^{p}_{m}(X)$ is a real-valued function satisfying the inequality $(\Delta_{b,m}+\alpha)u\geq 0$. Then $u\geq 0$.
\end{lem}
\begin{rem} \emph{The case $p=\infty$ is more complicated and involves the notion of stochastic completeness; see, for instance,
 \cite{Huang-11}, \cite{Keller-Lenz-10}, \cite{Keller-Lenz-09}.}
\end{rem}

In the remainder of this section and in section~\ref{SS:main-1-1}, we will use certain arguments of Section A
in~\cite{Kato86} and~\cite{Milatovic-10} in our setting. In the sequel, $\textrm{Ran }T$ denotes the range of an operator $T$.

\begin{lem}\label{L:hpmin-dense-range} Let $1<p<\infty$  and let $\lambda\in\mathbb{C}$ with $\textrm{Re }\lambda>0$. Then, $\textrm{Ran }(\hpmin+\lambda)$ is dense in $\lwp$.
\end{lem}
\begin{proof}
Let $u\in(\lpw)^*=\lqw$, be a continuous linear functional that annihilates
$(\lambda+H_{p,\min})\Gamma_{c}(X,F)$:
\begin{equation}\label{E:v-dual}
\sum_{x\in X} m(x)\langle(\lambda+\hpmin) v(x),u(x)\rangle_{F_x}=0,\qquad\textrm{for all
}v\in\Gamma_{c}(X,F).
\end{equation}
By assumption~(\ref{E:assumption-vcomp-lwp}) we know that $\widetilde{H}_{W,\Phi}v\in\lpw$. Since $u\in\lqw$, by Lemma~\ref{L:Green}(i) we have $u\in\widetilde{D}$. Now using Lemma~\ref{L:Green}(ii) in~(\ref{E:v-dual}), we get
\begin{equation}\nonumber
\sum_{x\in X} m(x)\langle v(x),(\overline{\lambda}+\widetilde{H}_{W^*,\Phi}) u(x)\rangle_{F_x}=0,\qquad\textrm{for all
}v\in \Gamma_{c}(X,F),
\end{equation}
where  $\overline{\lambda}$ is the complex conjugate of $\lambda$. The last equality leads to
\begin{equation}\label{E:prep-kato}
(\bar{\lambda}+ \Delta^{F,\Phi}_{b,m}+W^*)u=0.
\end{equation}
Using Kato's inequality~(\ref{E:kato-discrete}), assumption~(\ref{E:bdd-below-w}), and~(\ref{E:prep-kato}) we have
\begin{align*}
&|u|(\Delta_{b,m}|u|)\leq \textrm{Re}\,\langle\Delta^{F,\Phi}_{b,m}u,u\rangle_{F_{x}}\nonumber\\
&=-(\textrm{Re }\lambda)|u|^2-\textrm{Re}\,\langle W^*u,u \rangle_{F_x}\leq -(\textrm{Re }\lambda)|u|^2,
\end{align*}
where $|u|\in \ell_{m}^{p^*}(X)$ with $1<p^*<\infty$. Rewriting the last inequality, we obtain
\begin{equation}\nonumber
|u|(\Delta_{b,m}|u|+(\textrm{Re}\,\lambda)|u|)\leq 0.
\end{equation}
For all $x\in X$ such that $u(x)\neq 0$, we may divide both sides of the last inequality  by $|u(x)|$ to get
\begin{equation}\label{E:stochastic-completness}
(\Delta_{b,m}+\textrm{Re }\lambda)|u|\leq 0.
\end{equation}
Note that the inequality~(\ref{E:stochastic-completness}) also holds for those $x\in X$ such that $u(x)=0$; in this case, the left hand side of~(\ref{E:stochastic-completness}) is non-positive by~(\ref{E:ord-lap}). Thus, the inequality~(\ref{E:stochastic-completness}) holds for all $x\in X$.
By Lemma~\ref{L:keller-lenz-inequality}, from~(\ref{E:stochastic-completness}) we get $|u|\leq 0$. Hence, $u=0$.
\end{proof}
\subsection*{End of the Proof of Theorem~\ref{T:main-1}(i)}
The inequality~(\ref{E:accretive-hpmin}) means that $\hpmin$
is accretive in $\lpw$; see~(R1) in the Appendix with $j(u)=u|u|^{p-2}$. Hence, $\hpmin$ is
closable and $\overline{\hpmin}$ is accretive in $\lpw$; see the Appendix.
Therefore, for all $u\in\Dom(\overline{\hpmin})$ the following inequality holds:
\begin{equation}\label{E:closure-accretive-1}
\textrm{Re }\sum_{x\in X} m(x)\langle ({\hpmin} u)(x),u(x)|u(x)|^{p-2}\rangle_{F_{x}}\geq 0.
\end{equation}
Let $\lambda\in\mathbb{C}$ with $\textrm{Re } \lambda>0$. Using H\"older's inequality, from~(\ref{E:closure-accretive-1}) we get
\begin{equation}\label{E:hpmin-coercive}
(\textrm{Re }\lambda)\|u\|_{p}\leq \|(\lambda+\overline{\hpmin})u\|_{p},
\end{equation}
for all $u\in\Dom(\overline{\hpmin})$.
By Lemma~\ref{L:hpmin-dense-range} we know that $\textrm{Ran }(\hpmin+\lambda)$ is dense in $\lpw$. This, together with~(\ref{E:hpmin-coercive}), shows that
$\textrm{Ran }(\overline{\hpmin}+\lambda)=\lpw$. Hence, from~(\ref{E:hpmin-coercive}) we get
\[
\|(\xi+\overline{\hpmin})^{-1}\|\leq \frac{1}{\xi},\qquad\textrm{for all }\xi>0,
\]
where $\|\cdot\|$ is the operator norm $\lpw\to\lpw$. Thus, $-\overline{\hpmin}$ satisfies the conditions (C1), (C2) and (C3) of Hille--Yosida Theorem; see the Appendix. Hence, $-\overline{\hpmin}$ is the generator of a strongly continuous contraction semigroup on $\lpw$. $\hfill\square$
\subsection*{Proof of Theorem~\ref{T:main-1}(ii)}
Repeating the proof of Lemma~\ref{L:hpmin-dense-range} in the case $p=1$ and using Remark~\ref{R:stochastic-completeness},  from~(\ref{E:stochastic-completness}) with $u\in \Gamma_{\ell^{\infty}}(X,F)$ we obtain $|u|=0$. Therefore, for all $\lambda\in\mathbb{C}$ with $\textrm{Re } \lambda>0$, the set $\textrm{Ran }(H_{1,\min}+\lambda)$ is dense in $ \Gamma_{\ell^{1}_{m}}(X,F)$. From here on, we may repeat the proof
of Theorem~\ref{T:main-1}(i). $\hfill\square$

\section{Proof of Theorem~\ref{T:main-1-1}}\label{SS:main-1-1}
We begin with the following lemma.
\begin{lem}\label{L:closed-hpm} Let $1\leq p<\infty$ and $1/p+1/p^*=1$. Assume that~(\ref{E:assumption-vcomp-lwp-q}) is satisfied. Then $\hpm$ is a closed operator.
\end{lem}
\begin{proof} Let $u_k$ be a sequence of elements in $\Dom(\hpm)$ such that $u_k\to u$ and $\hpm u_k\to f$, as $k\to\infty$, using the norm convergence in $\lpw$. We need to show that $u\in\Dom(\hpm)$ and $f=\hpm u$. Let $v\in\Gamma_{c}(X,F)$ be arbitrary, and consider the sum
\[
\sum_{x\in X}m(x)\langle(\hpm u_k)(x),v(x)\rangle_{F_x}=\sum_{x\in X}m(x)\langle(\widetilde{H}_{W,\Phi} u_k)(x),v(x)\rangle_{F_x}.
\]
By Lemma~\ref{L:Green}(ii) we have
\begin{equation}\label{E:closed-norm-conv}
\sum_{x\in X}m(x)\langle(\widetilde{H}_{W,\Phi} u_k)(x),v(x)\rangle_{F_x}=\sum_{x\in X}m(x)\langle u_k(x),(\widetilde{H}_{W^*,\Phi}v)(x)\rangle_{F_x}.
\end{equation}
Using the norm convergence $u_k\to u$ in $\lpw$ and the assumption $\widetilde{H}_{W,\Phi}v\in\lqw$ with $1/p+1/p^*=1$, by H\"older's inequality we get
\[
\sum_{x\in X}m(x)\langle u_k(x),(\widetilde{H}_{W^*,\Phi}v)(x)\rangle_{F_x}\to \sum_{x\in X}m(x)\langle u(x),(\widetilde{H}_{W^*,\Phi}v)(x)\rangle_{F_x}.
\]
Using the norm convergence $\widetilde{H}_{W,\Phi}u_k\to f$ in $\lpw$, by H\"older's inequality we get
\[
\sum_{x\in X}m(x)\langle(\widetilde{H}_{W,\Phi} u_k)(x),v(x)\rangle_{F_x}\to \sum_{x\in X}m(x)\langle f(x),v(x)\rangle_{F_x}.
\]
Therefore, taking the limit as $k\to\infty$ on both sides of~(\ref{E:closed-norm-conv}), we obtain
\begin{equation}\label{E:closed-norm-conv-1}
\sum_{x\in X}m(x)\langle u(x),(\widetilde{H}_{W^*,\Phi}v)(x)\rangle_{F_x}=\sum_{x\in X}m(x)\langle f(x),v(x)\rangle_{F_x}.
\end{equation}
Since $u\in\lpw$ and since $\widetilde{H}_{W,\Phi}[\Gamma_{c}(X,F)]\subseteq \lqw$, we may use Lemma~\ref{L:Green}(i) to conclude $u\in \widetilde{D}$. Using Lemma~\ref{L:Green}(ii), we rewrite the left-hand side of~(\ref{E:closed-norm-conv-1}) as follows:
\begin{equation}\label{E:closed-norm-conv-2}
\sum_{x\in X}m(x)\langle u(x),(\widetilde{H}_{W^*,\Phi}v)(x)\rangle_{F_x}=\sum_{x\in X}m(x)\langle (\widetilde{H}_{W,\Phi}u)(x),v(x)\rangle_{F_x}.
\end{equation}
Since $v\in \Gamma_{c}(X,F)$ is arbitrary, by~(\ref{E:closed-norm-conv-1}) and~(\ref{E:closed-norm-conv-2}) we get $\widetilde{H}_{W,\Phi}u=f$. Thus, $u\in\Dom(\hpm)$ and $\hpm u=f$. Therefore, $\hpm$ is closed.
\end{proof}
\subsection*{Maximal Operator Associated with $\Delta_{b,m}$} Let $1\leq p<\infty$ and let $\Delta_{b,m}$ be as in~(\ref{E:ord-lap}). We define the maximal operator $L_{p,\max}$ in
$\lwp$ by the formula $L_{p,\max}u=\Delta_{b,m}u$ with the domain
\begin{equation}\label{E:domapm}\nonumber
\Dom(L_{p,\max})=\{u\in \lwp\cap \widetilde{D}: \Delta_{b,m}u\in \lwp\},
\end{equation}
where $\widetilde{D}$ is as in~(\ref{E:dom-F}) and sections are replaced by functions $X\to\mathbb{C}$.

Under the assumption (A1), it is known that $-L_{p,\max}$ generates
a strongly continuous contraction semigroup on $\lwp$ for all $1\leq p<\infty$; see Theorem 5 in~\cite{Keller-Lenz-09}. Thus, by Hille--Yosida Theorem (see the Appendix), we have
\begin{equation}\label{E:estimate-p-strichartz}
(0,\infty)\subset\rho(-L_{p,\max})\qquad\textrm{ and
}\qquad \|(\xi+L_{p,\max})^{-1}\|\leq
\frac{1}{\xi},
\end{equation}
for all $\xi>0$, where $\rho (T)$ denotes the resolvent set of an operator $T$.
\begin{lem}\label{L:kato-2} Let $1\leq p<\infty$ and let
$\lambda\in\mathbb{C}$ with $\textrm{Re }\lambda>0$. Assume that the hypotheses (A1) and~(\ref{E:bdd-below-w}) are satisfied. Then, the following
properties hold:
\begin{enumerate}
\item [(i)] for all $u\in\Dom(\hpm)$, we have
\begin{equation}\label{E:kato-to-prove}
(\textrm{Re }\lambda)\|u\|_{p}\leq \|(\lambda+\hpm)u\|_{p};
\end{equation}
\item [(ii)] the operator $\lambda+ \hpm\colon \Dom(\hpm)\subset
\lpw\to\lpw$ is injective.
\end{enumerate}
\end{lem}
\begin{proof} Let
$u\in\Dom(\hpm)$ and $f:=(\lambda+\hpm)u$. By the definition of
$\Dom(\hpm)$, we have $f\in \lpw$, where $1<p<+\infty$.
Using~(\ref{E:kato-discrete}) and~(\ref{E:bdd-below-w}) we get
\begin{align*}
&|u|((\textrm{Re}\,\lambda+\Delta_{b,m})|u|)\leq \textrm{Re}\,\langle(\lambda+\Delta^{F,\Phi}_{b,m})u,u\rangle_{F_{x}}\nonumber\\
&\leq\textrm{Re}\,\langle(\lambda+\Delta^{F,\Phi}_{b,m}+W)u,u\rangle_{F_{x}}=\textrm{Re}\,\langle f,u\rangle_{F_x}\leq |f||u|.
\end{align*}
In what follows, we denote $\xi:=\textrm{Re}\,\lambda$.
For all $x\in X$ such that $u(x)\neq 0$, we may divide both sides of the last inequality  by $|u(x)|$ to get
\begin{equation}\label{E:kato-1-intro}
(\xi+ \Delta_{b,m})|u| \ \leq \ |f|.
\end{equation}
Note that the inequality~(\ref{E:kato-1-intro}) also holds for those $x\in X$ such that $u(x)=0$; in this case, the left hand side of~(\ref{E:kato-1-intro}) is non-positive by~(\ref{E:ord-lap}). Thus, the inequality~(\ref{E:kato-1-intro}) holds for all $x\in X$.

According to~(\ref{E:estimate-p-strichartz}) the linear operator
\[
(\xi+L_{p,\max})^{-1}\colon \lwp \to \lwp
\]
is  bounded. Hence, we can rewrite~(\ref{E:kato-1-intro}) as
\begin{equation}\label{E:temp-1-kato}
(\xi + \Delta_{b,m})[(\xi+L_{p,\max})^{-1}|f| -|u|] \geq 0.
\end{equation}
Since
\[
(\xi+L_{p,\max})^{-1}|f|\in \lwp \qquad\textrm{ and }\qquad|u|\in\lwp,
\]
it follows that $((\xi+L_{p,\max})^{-1}|f| -|u|)\in \lwp$.  Hence, applying Lemma~\ref{L:keller-lenz-inequality} to~(\ref{E:temp-1-kato}) we get
\begin{equation}\label{E:kato-link-strichartz}\nonumber
|u|\leq (\xi+L_{p,\max})^{-1}|f|.
\end{equation}
Taking the $\ell^p$-norms on both sides and using~(\ref{E:estimate-p-strichartz})
we get
\[
\|u\|_{p}\leq
\|(\xi+L_{p,\max})^{-1}|f|\|_{p}\leq\frac{1}{\xi}\|f\|_{p},
\]
and~(\ref{E:kato-to-prove}) is proven.
We turn to property (ii). Assume that $u\in \Dom(\hpm)$ and
$(\lambda+\hpm)u=0$. Using~(\ref{E:kato-to-prove}) we
get $\|u\|_{p}=0$, and hence $u=0$. This shows that $\lambda+\hpm$
is injective.
\end{proof}

\subsection*{End of the Proof of Theorem~\ref{T:main-1-1}} We will consider the cases $1<p<\infty$ and $p=1$ simultaneously, keeping in mind the stochastic completeness assumption on $(X,b,m)$ when $p=1$.
Since $\hpmin\subset\hpm$ and
since $\hpm$ is closed (see Lemma~\ref{L:closed-hpm}), it follows that
$\overline{\hpmin}\subset\hpm$. To prove the equality $\overline{\hpmin}=\hpm$, it is enough to show that $\Dom(\hpm)\subset \Dom(\overline{\hpmin})$. Let $\xi>0$, let $u\in\Dom(\hpm)$, and consider
\begin{equation}\label{E:element-v}
v:=(\overline{\hpmin}+\xi)^{-1}(\hpm+\xi)u.
\end{equation}
By Theorem~\ref{T:main-1}, the element $v$ is well-defined, and $v\in\Dom(\overline{\hpmin})$.

Since $\overline{\hpmin}\subset\hpm$, from~(\ref{E:element-v}) we get
\[
(\hpm+\xi)(v-u)=0.
\]
Since $\hpm+\xi$ is an injective operator (see Lemma~\ref{L:kato-2}), we get $v=u$. Therefore, $u\in\Dom(\overline{\hpmin})$. $\hfill\square$

\section{Proof of Theorem~\ref{T:main-5}}\label{S:t-5}
The following lemma, whose proof is given in Proposition 4.1 of~\cite{vtt-11-4}, describes an important property of regular graphs. For the case of metrically complete graphs, see~\cite{HKMW}.

\begin{lem}\label{L:finsup} Assume that $(X,b,m)$ is a locally finite graph with a path metric $d_{\sigma}$. Additionally, assume that $(X,b,m)$ is regular in the sense of Definition~\ref{ikx}. Let $X_{\varepsilon}$ be as in~(\ref{E:x-eps}). Then, closed and bounded subsets of
$X_{\varepsilon}$ are finite.
\end{lem}
By Remark~\ref{R:min-def-op} and Lemma~\ref{L:Green}(ii), $\widetilde{H}_{W,\Phi}|_{\Gamma_{c}(X,F)}$ is a symmetric operator in $\lw$. To prove Theorem~\ref{T:main-5} we follow the method of Theorem 1.5 in~\cite{Milatovic-Truc-13}, which goes back to~\cite{Col-Tr} in the continuous setting. The main ingredient is the following Agmon-type estimate:

\begin{lem}\label{L:Hor}
Let $\lambda\in\mathbb{R}$ and let $v\in\lw$ be a weak solution of $(\widetilde{H}_{W,\Phi}-\lambda)v=0$.
Assume that  there exists a constant $c_1>0$ such that, for all $u \in \Gamma_{c}(X,F)$
\begin{equation}\label{E:bou}
( u,\, (\widetilde{H}_{W,\Phi}-\lambda) u )  \geq
 \dfrac{1}{2}\sum_{x\in X}\max \left(\dfrac{1}{D(x)^2},1
\right) m(x) |u(x)|^2_{F_x}  +  c_1 \|u\|^2,
\end{equation}
where $D(x)$ is as in~(\ref{E:dist-boundary}). Then $v\equiv 0$.
\end{lem}

\begin{proof}
Let $\rho$ be a number such that $0< \rho < 1/2$.
For any $\varepsilon >0$, we define $f_{\varepsilon}\colon X \rightarrow \mathbb{R}$
by $f_{\varepsilon}(x)=F_{\varepsilon}(D(x))$, where $D(x)$ is as in~(\ref{E:dist-boundary}) and $F_{\varepsilon}\colon{\mathbb{R}}^{+} \rightarrow \mathbb{R}$ is given by $F_{\varepsilon}(s)=0$ for $s\leq \varepsilon$; $F_{\varepsilon}(s)=(s-\varepsilon)/(\rho  - \varepsilon)$ for
$\varepsilon\leq s \leq  \rho$; $F_{\varepsilon}(s)=s$ for $\rho  \leq s \leq  1$; $F_{\varepsilon}(s)=1$ for $s\geq 1$.

Let us fix a vertex $x_0$. For any $\alpha >0$, we define
$g_\alpha\colon X \rightarrow \mathbb{R}$ by $g_{\alpha}(x) =G_\alpha(d_\sigma(x_0,x))$, where  $G_\alpha\colon {\mathbb{R}}^{+} \rightarrow \mathbb{R}$ is given by $G_\alpha(s)=1$ for  $s\leq 1/\alpha$; $G_\alpha(s)=-\alpha  s + 2$ for $1/\alpha  \leq s \leq  2/\alpha$; $G_\alpha(s)=0$ for $s\geq 2/\alpha$. We also define
\begin{equation}\label{E:support}\nonumber
 E_{\varepsilon , \alpha}:= \{x \in X \colon \varepsilon \leq D(x) \  {\rm and}  \ d_\sigma(x_0,x) \leq 2/\alpha \}.
\end{equation}
By Lemma~\ref{L:finsup} the set $E_{\varepsilon, \alpha}$ is finite because $E_{\varepsilon, \alpha}$ is a closed
and bounded subset of $X_{\varepsilon}$, where $X_{\varepsilon}$ is as in~(\ref{E:x-eps}). Since the support of $f_{\varepsilon}g_{\alpha}$ is contained in $E_{\varepsilon, \alpha}$, it follows that $f_{\varepsilon}g_{\alpha}$ is finitely supported.
Using Lemma 4.1 in~\cite{vtt-11-2} it is easy to see that $f_{\varepsilon}g_\alpha$ is a $\beta$-Lipschitz function with respect to $d_{\sigma}$,  where $\beta={\rho}/({\rho-\varepsilon})+\alpha$. By Lemma~\ref{L:lemma-ground-state} with with $g$ replaced by $f_{\varepsilon}g_\alpha$, unitarity of $\Phi_{y,x}$, $\beta$-Lipschitz property of $f_{\varepsilon}g_\alpha$, and Defintion~\ref{D-intrinsic}, we have
\begin{equation}\label{E:equ-rhs}
(f_{\varepsilon}g_{\alpha} v,\, (\widetilde{H}_{W,\Phi}-\lambda) (f_{\varepsilon}g_{\alpha} v))
\leq \frac{1}{2}\left(\frac{\rho}{\rho-\varepsilon}+\alpha\right)^2\sum _{x\in X} m(x) |v(x)|^2_{F_x}.
\end{equation}

On the other hand, by the definitions of $f_{\varepsilon}$ and $g_{\alpha}$ and the assumption~(\ref{E:bou}) we have
\begin{equation} \label{E:equ-lhs}
(f_{\varepsilon}g_{\alpha}v,\,(\widetilde{H}_{W,\Phi}-\lambda) (f_{\varepsilon} g_{\alpha}v))
\geq \frac{1}{2} \sum _{x\in S_{\rho , \alpha}} m(x)|v(x)|^2_{F_x} +c_1 \| f_{\varepsilon}g_{\alpha} v
\|^2,
\end{equation}
where
\[
S_{\rho, \alpha}:=\{x \in X \colon \rho \leq D(x) \  {\rm and}  \ d_\sigma(x_0,x) \leq 1/\alpha\}.
\]
Combining (\ref{E:equ-lhs}) and (\ref{E:equ-rhs}) we obtain
\[
\frac{1}{2} \sum _{x\in S_{\rho , \alpha}}m(x)|v(x)|^2_{F_x} +c_1 \| f_{\varepsilon} g_{\alpha}v\|^2\leq \frac{1}{2}\left(\frac{\rho}{\rho-\varepsilon}+\alpha\right)^2\sum _{x\in X}m(x) |v(x)|^2_{F_x}.
\]
We fix $\rho$ and $\varepsilon$, and let $\alpha \to 0+$. After that, we let  $\varepsilon \to 0+$. Finally, we take the limit as $\rho\to 0+$.
As a result, we get $v\equiv 0$.
\end{proof}

\subsection*{End of the Proof of Theorem~\ref{T:main-5}} Since  $\Delta^{F,\Phi}_{b,m}|_{\Gamma_{c}(X,F)}$ is a non-negative operator, for all $u \in \Gamma_{c}(X,F)$, we have
\[
(u,\, \widetilde{H}_{W,\Phi} u) \geq \sum_{x\in X} m(x) \langle W(x)u(x),u(x)\rangle_{F_{x}}.
\]
Therefore, using assumption~(\ref{E:potential-minorant}) we obtain:
\begin{align}\label{E:bou-new}
&(u,\, (\widetilde{H}_{W,\Phi}-\lambda) u)\geq \dfrac{1}{2}\sum_{x\in X} \frac{1}{D(x)^2} m(x)|u(x)|^2_{F_{x}}-(\lambda+C)\|u\|^2\nonumber\\
&\geq  \dfrac{1}{2}\sum_{x\in X} \max\left(\frac{1}{D(x)^2},1\right) m(x)|u(x)|^2_{F_{x}} -(\lambda+C+1/2)\|u\|^2.
\end{align}
Choosing, for example, $\lambda=-C-3/2$ in~(\ref{E:bou-new}) we get the
inequality (\ref{E:bou}) with $c_1=1$. Thus, $(\widetilde{H}_{W,\Phi}-\lambda)|_{\Gamma_{c}(X,F)}$ with $\lambda=-C-3/2$ is a symmetric operator satisfying
$(u,\, (\widetilde{H}_{W,\Phi}-\lambda) u)\geq \|u\|^2$, for all $u\in \Gamma_{c}(X,F)$. By Theorem X.26 in~\cite{rs} we know that the essential self-adjointness of $(\widetilde{H}_{W,\Phi}-\lambda)|_{\Gamma_{c}(X,F)}$ is equivalent to the following statement: if $v\in\lw$ satisfies $(\widetilde{H}_{W,\Phi}-\lambda)v=0$, then $v=0$. Thus, by Lemma~\ref{L:Hor}, the operator $(\widetilde{H}_{W,\Phi}-\lambda)|_{\Gamma_{c}(X,F)}$ is essentially self-adjoint. Thus, $\widetilde{H}_{W,\Phi}|_{\Gamma_{c}(X,F)}$ is essentially self-adjoint. $\hfill\square$

\section*{Appendix}
In this section we review some concepts from the theory of one-parameter semigroups of operators on Banach spaces. Our exposition follows Chapters I and II of~\cite{engel-nagel}. A family of bounded linear operators $(T(t))_{t\geq0}$ on a Banach
space $\mathscr{X}$ is called a \emph{strongly continuous semigroup} (or
$C_0$-\emph{semigroup}) if it satisfies the functional equation
\[
T(t + s) = T(t)T(s),\quad\textrm{for all }t, s\geq 0,\qquad T(0) = I,
\]
and the maps $t\mapsto T(t)u$ are continuous from $\mathbb{R}_{+}$ to $\mathscr{X}$ for all $u\in \mathscr{X}$. Here, $I$ stands for the identity operator on $\mathscr{X}$.

The \emph{generator} $A \colon \Dom(A) \subset \mathscr{X} \to \mathscr{X}$ of a strongly continuous
semigroup $(T(t))_{t\geq0}$ on a Banach space $\mathscr{X}$ is the operator
\[
Au:=\lim_{h\to 0+}\frac{T(h)u-u}{h}
\]
defined for every $u$ in its domain
\[
\Dom(A) := \{u \in \mathscr{X}: \lim_{h\to 0+}h^{-1}(T(h)u-u)\textrm{ exists}\}.
\]
By Theorem II.1.4 in~\cite{engel-nagel}, the generator of a strongly continuous
semigroup is a closed and densely defined operator that determines the semigroup uniquely.

A linear operator $A$ on a Banach space $\mathscr{X}$ with norm $\|\cdot\|$ is called \emph{accretive}
if
\[
\|(\xi+A)u\|\geq \xi\|u\|,
\]
for all $\xi>0$ and all $u\in\Dom(A)$. In the literature on semigroups of operators, the term \emph{dissipative} is used when referring to an operator $A$ such that $-A$ is accretive. If $A$ is a densely defined accretive operator, then $A$ is closable and its closure $\overline {A}$ is also accretive; see Proposition II.3.14 in~\cite{engel-nagel}.

We now give another description of accretivity. Let $\mathscr{X}^{*}$ be the dual
space of $\mathscr{X}$. By the Hahn–-Banach theorem, for every $u\in \mathscr{X}$ there exists
$u^*\in \mathscr{X}^*$ such that $\langle u,u^*\rangle=\|u\|^2=\|u^*\|^2$, where $\langle u,u^*\rangle$ denotes the evaluation of the functional $u^*$ at $u$.
For every $u \in\mathscr{X}$, we define
\[
\mathscr{J}(u) :=\{u^*\in \mathscr{X}^*:\langle u,u^*\rangle=\|u\|^2=\|u^*\|^2\}.
\]
By Proposition II.3.23 of~\cite{engel-nagel}, an operator $A$ is accretive if and only if for
every $u\in\Dom (A)$ there exists $j(u)\in\mathscr{J}(u)$ such that
\[
\textrm{Re }\langle Au,j(u)\rangle\geq 0. \eqno{{\rm { (R1)}}}
\]
An operator $A$ on a Banach space $\mathscr{X}$ is called \emph{maximal accretive} if it is accretive and $\xi+A$ is surjective for all $\xi>0$. There is a connection between maximal accretivity and self-adjointness of operators on Hilbert spaces: $A$ is a self-adjoint and non-negative operator if and only if $A$ is symmetric, closed, and maximal accretive; see Problem V.3.32 in~\cite{Kato80}.

A contraction semigroup $(T(t))_{t\geq0}$ on a Banach space $\mathscr{X}$ is a semigroup such that $\|T(t)\|\leq 1$ for all $t\geq 0$, where $\|\cdot\|$ denotes the operator norm (of a bounded linear) operator $\mathscr{X}\to \mathscr{X}$. Generators of strongly continuous contraction semigroups are characterized as follows (Theorem II.3.5 in~\cite{engel-nagel}):

\medskip

\noindent{\textbf{Hille--Yosida Theorem.}} \emph{An operator $A$ on a Banach space generates a strongly continuous contraction semigroup  if and only if the following three conditions are satisfied:
\begin{enumerate}
\item [(C1)] $A$ is densely defined and closed;
\item [(C2)] $(0,\infty)\subset\rho(A)$, where $\rho(A)$ is the resolvent set of $A$;
\item [(C3)] $\|(\xi-A)^{-1}\|\leq\xi^{-1}$, for all $\xi>0$.
\end{enumerate}}

\medskip

Finally, we note that if $A$ generates a strongly continuous contraction semigroup, then $-A$ is maximal accretive.

\section*{Acknowledgment} The authors are grateful to the anonymous referee
for providing valuable suggestions and helping us improve the presentation of the material.


\end{document}